\documentclass[journal,twoside]{IEEEtran}

\usepackage[utf8]{inputenc}
\usepackage[T1]{fontenc}
\usepackage{hyperref}
\usepackage{url}
\usepackage{booktabs}
\usepackage{amsfonts}
\usepackage{nicefrac}
\usepackage{microtype}
\usepackage{graphicx}
\graphicspath{{media/}}
\usepackage[table,dvipsnames]{xcolor}
\usepackage{multirow}
\usepackage{amsmath,amssymb}
\usepackage{amsthm}
\usepackage{mathrsfs}
\usepackage{textcomp}
\usepackage{algorithm}
\usepackage{algorithmicx}
\usepackage{algpseudocode}
\usepackage{braket}
\usepackage{listings}
\usepackage{mathtools}
\usepackage[caption=false]{subfig}

\newtheorem{theorem}{Theorem}
\newtheorem{definition}{Definition}  

\title{BBQ-mIS: a parallel quantum algorithm for graph coloring problems}

\author{\IEEEauthorblockN{Chiara Vercellino\IEEEauthorrefmark{1}\IEEEauthorrefmark{2}\IEEEauthorrefmark{3},
Giacomo Vitali\IEEEauthorrefmark{1}\IEEEauthorrefmark{2},
Paolo Viviani\IEEEauthorrefmark{1},
Edoardo Giusto\IEEEauthorrefmark{2},\\
Alberto Scionti\IEEEauthorrefmark{1},
Andrea Scarabosio\IEEEauthorrefmark{1},
Olivier Terzo\IEEEauthorrefmark{1},
Bartolomeo Montrucchio\IEEEauthorrefmark{2}}\\
\IEEEauthorrefmark{1}\textit{LINKS Foundation}, Torino, Italy \\
\IEEEauthorrefmark{2}\textit{\textit{DAUIN}, Politecnico di Torino}, Torino, Italy\\
\IEEEauthorrefmark{3}\textit{chiara.vercellino@linksfoundation.com}
}

\begin{document}
\maketitle

\begin{abstract}
Among the limitations of current quantum machines, the qubits count represents one of the most critical challenges for porting reasonably large computational problems, such as those coming from real-world applications, to the scale of the quantum hardware.
In this regard, one possibility is to decompose the problems at hand and exploit parallelism over multiple size-limited quantum resources. To this purpose, we designed a hybrid quantum-classical algorithm, \textit{i.e.}, \textbf{BBQ-mIS}, to solve graph coloring problems on Rydberg atoms quantum machines. 
The \textbf{BBQ-mIS} algorithm combines the natural representation of \textit{Maximum Independent Set} (\textit{MIS}) problems onto the machine Hamiltonian with a \textit{Branch\&Bound} (BB) approach to identify a proper graph coloring.
In the proposed solution, the graph representation emerges from qubit interactions (qubits represent vertexes of the graph), and the coloring is then retrieved by iteratively assigning one color to a maximal set of independent vertexes of the graph, still minimizing the number of colors with the \textit{Branch\&Bound} approach.
We emulated real quantum hardware onto an IBM Power9-based cluster, with
32 cores/node and 256 GB/node, and exploited an MPI-enhanced library to implement the parallelism for the \textbf{BBQ-mIS} algorithm. Considering this use case, we also identify some technical requirements and challenges for an effective HPC-QC integration. 
The results show that our problem decomposition is effective in terms of graph coloring solutions quality, and provide a reference for applying this methodology to other quantum technologies or applications.

\end{abstract}


\section{Introduction}\label{sec:intro}

Recent years have seen a steep rise in interest concerning quantum computers and their integration with High Performance Computing (HPC). For this purpose, it is important to identify both infrastructures to make the classical and quantum machines communicate and applications that can benefit from such an integrated system.



Concerning this last aspect, we propose a hybrid-quantum classical approach to solve graph coloring (GC) problems that emerge from several industrial applications, \textit{e.g.}, optimal deployment of communication networks \cite{barillaro2023comparison} and analysis of biological networks \cite{khor2010application}. The key idea in our algorithm, \textbf{BBQ-mIS}, is to combine a classical Branch\&Bound algorithm, which by definition is highly parallelizable, with a quantum routine to sample feasible (partial-)coloring solutions.

Among the various quantum technologies, neutral atoms machines~\cite{grimm2000optical} are extremely suitable for approaching graph combinatorial optimization problems, such as GC.
They employ Rydberg atoms, \textit{i.e.} neutral atoms, to act as \emph{qubits}, organizing them on a 2D/3D register. Then, the interactions among these atoms, subject to laser pulses \cite{pasqal_optim}, generate the machine Hamiltonian $H$ (\textit{e.g.}, Ising)

\begin{equation}\label{hamiltonian_eq}
    H= \sum_{i=1}^{n} \frac{\hslash \Omega}{2} \sigma_{i}^{x} - \sum_{i=1}^{n} \frac{\hslash \delta}{2} \sigma_{i}^{z} + \sum_{j>i} \frac{C_6}{d_{ij}^{6}} n_i n_j
\end{equation}
where, $n_i = (1+\sigma_i^z)/2$ is the Rydberg state occupancy, $\sigma^{x,z}_i$ are the $i$-th qubit's Pauli matrices, $\Omega$ and $\delta$ are respectively the Rabi frequency and the detuning of the controlling laser system, and $d_{ij}$ is the Euclidean distance between qubits $i$ and $j$ in the register.
Hence, the way the machine Hamiltonian evolves depends on the laser pulses parameters and on the positions of qubits (neutral atoms), which assume, once measured, one out of two possible quantum states (\textit{i.e.}, the excited Rydberg state $\vert 1 \rangle$ or the ground state $\vert 0 \rangle$). From this derives the connection with binary optimization variables.


There is also another effect of quantum mechanics that allows the representation of graphs directly on the quantum machine register, \textit{i.e.}, the \textit{blockade} effect.
It acts as a threshold on the qubit-to-qubit distance, discriminating the closer ones from the farther ones~\cite{ciampini2015ultracold}.
This threshold distance is the critical one at which the strength of the qubit-to-qubit interactions balances with the Rabi frequency of the laser pulses~\cite{picken2018entanglement}, and it is called \textit{blockade radius} $r_b = (\frac{C_6}{\hslash\Omega})^{1/6}$.
In this way, the interactions between qubits, which can be assimilated to the vertexes of a graph, can induce a Unit-Disk Graph (UDG)~\cite{CLARK1990165}: qubit positions in the quantum register are UDG's vertex positions, and edges in the UDG link two vertexes whenever their Euclidean distance is shorter than the blockade radius.

Moreover, the blockade effect prevents qubits, which fall within the blockade radius, from being both in the excited state $|1\rangle$. Thus, retrieving the ground state of $H$ coincides with computing the largest set of non-interacting qubits on the register, \textit{i.e.}, the \textit{Maximum Independent Set} (\textit{MIS}) of the corresponding UDG.




However, not all the combinatorial optimization problems can be easily reformulated as \textit{MIS} problems on UDGs, or, when this is the case \cite{embedding_quera}, the problem mapping onto the quantum hardware scales badly in the number of physical qubits to represent the binary variables. Willing to solve the problem of interest and, at the same time, to reduce as much as possible the needed qubits, we designed the \textbf{BBQ-mIS} algorithm for GC problems.

More in detail, starting from a UDG representation of the target graph thanks to the methodology presented in \cite{embedding_links}, the proposed quantum-enhanced heuristic solves the graph coloring problem by iteratively solving \textit{MIS} problems. One could notice that being able to retrieve a UDG from the graph of interest might limit the generality of the approach. However, it is still possible to map one logical qubit into multiple physical qubits to increase the UDG's connectivity \cite{embedding_quera} or rely on an approximated UDG (an approximation of the original graph with some missing edges). In fact, our approach does not need to find the optimal \textit{MIS} to solve GC problems.

\textbf{BBQ-mIS} successfully solves graph coloring problems, for all the graph samples of our dataset, with the minimum number of colors. The benchmark solutions are provided by the exact solver \textit{Gurobi} \cite{bixby2007gurobi}.


\section{Related Work}

Hybrid quantum-classical algorithms keep on emerging in several research areas. They range from the optimization field, \textit{e.g.}, for dynamic portfolio optimization with minimal holding period exploiting \textit{D-Wave 2000Q} processor \cite{mugel2021hybrid}, or, in a sort of complementary direction, the design of optimization methods for parameterized quantum circuits \cite{nakanishi2020sequential, graham2022multi}, to machine learning, \textit{e.g.}, deep learning time series for ship motion forecasting \cite{li2022ship}, to material simulation using variational quantum algorithms to determine the ground state of molecular problems by combining classical and quantum (neutral atom machines) hardware \cite{deKeijzer2023pulsebased}.

Thus, due to the extensive literature on the topic, we focus on the quantum-classical algorithms to solve the \textit{Graph Coloring} (\textit{GC}) problem.

The \textit{GC} problem admits a quadratic unconstrained binary optimization (QUBO), so it is possible to solve the corresponding combinatorial optimization problem through quantum annealing. For instance, C. Silva et al.\cite{silva2020mapping} compared the results of solving the QUBO problem of \textit{GC} both with a fully-classical simulated annealing method and using a \textit{D-Wave} quantum machine; according to their study, the quantum approach found more solutions.
The paper \cite{wikeckowski2019disorder} proposes a similar approach (they still exploit the mapping of the combinatorial optimization problem into chimera graphs and use quantum annealers) but focus more on the effect of an external noise source, which is a hot topic for noisy intermediate-scale quantum devices. In particular, they manage to exploit the noise-generated static disorder to improve performance. 
Quantum annealing methods for \textit{GC} are also discussed in \cite{titiloye2012parameter, kudo2018constrained, titiloye2011quantum}, but in this case, QUBO formulation is not adopted, constraints and objective function are directly encoded into the problem and driving Hamiltonians.

Quantum circuits for \textit{GC} problems are the topic of \cite{do2020planning}. In this work, M. Do et al. extend a previous study on Max-Cut problems \cite{mohar1990eigenvalues}: they implement the Quantum Alternating Operator Ansatz (QAOA) for \textit{GC} on quantum processor architecture. Together with the circuit routing, they propose a qubit initialization approach that allows for shorter makespan compilation.
In \cite{tabi2020quantum}, the authors designed a space-efficient quantum optimization algorithm for the \textit{GC} problem circuit implementation. They approached the problem through a gate-based implementation and showed that gain in the number of logical qubits, which are exponentially reduced in the number of colors, comes to a cost of deeper circuits.

The paper \cite{wang2011improved} presents a generalization of the Grover algorithm to operate on ternary quantum circuits: the augmented representation capability enhanced by the use of \textit{qudits} instead of \textit{qubits} allows the reduction of the circuit complexity, the simplification of the experimental setup and the enhancement of the algorithm efficiency.
Multi-level computational units, specifically \textit{qutrits}, are also present in the work of S. Bravyi et al. \cite{Bravyi2022hybridquantum}. In this case, increasing the state space enables a $3$- coloring formulation for \textit{GC} problems.

Finally, the authors in \cite{ardelean2022graph} introduce quantum circuits implementing genetic algorithms (GA) to solve NP-hard optimization problems. They apply their methodology to \textit{GC} problems and examine the results obtained with a \textit{Qiskit} simulation environment.

\section{Methodology}

In this section, we illustrate two hybrid quantum-classical algorithms. The first one, \textbf{Greedy-it-MIS}, is a quantum-enhanced version of a classical heuristic for graph coloring problems; the other one, \textbf{BBQ-mIS}, is a novel approach that exploits \textit{MIS} solution to obtain a feasible graph coloring, still minimizing the number of colors used, thanks to a Branch\&Bound approach.

Before diving into the descriptions of \textbf{Greedy-it-MIS} and \textbf{BBQ-mIS} algorithms, we introduce some basic notions of graph theory, thus setting also the notation.

Given a graph $\mathcal{G}=(\mathcal{V},\mathcal{E})$, $\mathcal{V}$ is the set of \textit{vertexes}, $n=|\mathcal{V}|$, $\mathcal{E}$ is the set of undirected \textit{edges}, and $A$ the adjacency matrix of $\mathcal{G}$. A feasible \textit{coloring} of $\mathcal{G}$ consists of assigning to each vertex in $\mathcal{V}$ a color such that vertexes that share an edge have different colors. A \textit{graph coloring} (\textit{GC}) problem \cite{jensen2011graph} arises when the feasible coloring of $\mathcal{G}$ is targeted along with the minimization of colors. The number of colors that solve a \textit{GC} problem is called the \textit{chromatic number} of $\mathcal{G}$, and it is denoted by $\chi(\mathcal{G})$.
Another well-known graph combinatorial optimization problem is the \textit{Maximum Independent Set} (\textit{MIS}) problem. It consists in finding the largest independent set \cite{tarjan1977finding} of $\mathcal{V}$ (see def. \ref{def:is}).

\begin{definition}[Independent set]\label{def:is}
Given a graph $\mathcal{G}=(\mathcal{V},\mathcal{E})$, an \textit{independent set} is a set of vertexes such that no two vertexes share an edge in $\mathcal{E}$.
\end{definition}

\subsection{Greedy-it-MIS for Graph Coloring problems}

A first approach to solving \textit{GC} problems, without requiring one optimization variable, \textit{i.e.} a logical qubit, for each possible vertex color assignment \cite{glover2018tutorial}, and exploiting the straightforward solution of \textit{MIS} problems on neutral atoms, follows from a simple argument. An \textit{MIS} solution is a feasible color assignment: provided that all vertexes in the \textit{MIS} are independent, they can be colored with the same color. Then, the procedure can iterate over the induced graph $\mathcal{G}'=(\mathcal{V}\setminus MIS, \mathcal{E}')$ that is the original graph $\mathcal{G}$ after removing the vertexes contained in the \textit{MIS} and the corresponding edges. This procedure iterates until all vertexes in $\mathcal{V}$ are associated with a color, as shown in algorithm \ref{alg:greedy}. We named this coloring strategy \textbf{Greedy-it-MIS} since we consider the first solution coming from the \textit{MIS solver}, after checking for its independence, to obtain $\mathcal{G}'$, without evaluating the impact of choosing that particular solution on the overall coloring. In a typical short-sighted way, \textbf{Greedy-it-MIS} looks just one step at a time, hence the name \textit{Greedy}. As a result, \textbf{Greedy-it-MIS} provides feasible \textit{GC} solutions but does not target the minimization of the colors.

\begin{algorithm}
\caption{\textbf{Greedy-it-MIS} for Graph Coloring problems}\label{alg:greedy}
\begin{algorithmic}
\Ensure $C$ is a feasible coloring for graph $\mathcal{G}=(\mathcal{V},\mathcal{E})$
\State $k \gets 0$ \Comment{number of colors used}
\While{$|\mathcal{V}|>0$}
\If{$|\mathcal{E}|>0$}  \Comment{$\mathcal{V}$ is not a MIS}
    \State $MIS \gets MIS\_solver(\mathcal{G}) $
    \State $\mathcal{G}=(\mathcal{V},\mathcal{E}) \gets \mathcal{G}'=(\mathcal{V}\setminus MIS, \mathcal{E}')$
\Else
    \State $MIS \gets \mathcal{V}$
    \State $\mathcal{G}=(\mathcal{V},\mathcal{E}) \gets \emptyset$
\EndIf
\State $k \gets k+1$
\State $C_k \gets MIS$
\EndWhile
\end{algorithmic}
\end{algorithm}

\subsection{BBQ-mIS for Graph Coloring problems}

Willing to improve the solutions to GC problems, one should identify some exploitable directions that \textbf{Greedy-it-MIS} does not take into account.

First of all, the probabilistic behavior of quantum systems requires multiple measurements to provide an \textit{MIS} solution. These repeated experiments return a histogram of possible solutions, where the most occurring one is selected as the \textit{MIS} solution by the \textbf{Greedy-it-MIS}. However, other solutions in the histogram can yield valuable information, and they come for free by the standard measurement procedure. As a matter of fact, there is no guarantee that the solution to an \textit{MIS} problem is unique, so even when the optimal coloring contains a maximum independent set, it might not be the first one given by the \textit{MIS} solver. Moreover, it happens that optimal colorings do not include \textit{MIS} solutions 
so, we should also consider smaller independent sets from the histogram of solutions.


Nevertheless, it would be computationally hard to explore by brute-force all possible independent set combinations provided by an \textit{MIS} solver. It is necessary to have a way of cutting out some solutions and evaluating the impact on the targeted GC problem of choosing one independent set over another one. An essential building block to deal with this issue is Theorem \ref{th:mIS} \cite{de_Lima_Carmo_2018, christofides1971algorithm}.

\begin{definition}[Maximal independent set]
 A \textit{maximal Independent Set} (\textit{mIS}) of $\mathcal{G}=(\mathcal{V},\mathcal{E})$ is an independent set, which is not properly contained in another independent set of $\mathcal{G}$, \textit{i.e.}, adding one or more vertexes to a maximal independent set would result in losing its independence.
\end{definition}

\begin{theorem}[Optimal graph coloring]\label{th:mIS}
Every graph $\mathcal{G}=(\mathcal{V},\mathcal{E})$ has an optimal coloring in which (at least) one of the colors is a \textit{maximal independent set}.
\end{theorem}

\textit{Proof}:
$C=\{ C_1, \ldots , C_k \}$ vertex sets defining the optimal coloring of $\mathcal{G}=(\mathcal{V},\mathcal{E})$, $k$ colors. Let $I$ be an \textit{mIS} of $\mathcal{G}$ s.t. $I$ contains $C_1$ then $C' = \{ I, C_2\setminus I,\ldots, C_k\setminus I \}$ is still an optimal coloring.

Thanks to this theorem, we can restrict our solutions space investigation to \textit{mIS}s.

At this point, we can set up an effective minimization of colors for a \textit{GC} problem. The proposed solution is a Branch\&Bound (BB) approach that branches on \textit{mIS} solutions and considers as bounds the lower bounds on the chromatic number: \textbf{BBQ-mIS}.

More precisely, \textbf{BBQ-mIS} leverages well-known lower bounds on $\chi(\mathcal{G})$ for the pruning criteria, since the \textit{GC} is a minimization problem, and exploits upper bounds on $\chi(\mathcal{G})$ for the exploration policy of the BB tree. The considered bounds are valid for any graph $\mathcal{G}=(\mathcal{V},\mathcal{E})$.

We denote with $\Delta(\mathcal{G})$ the maximum vertex degree in $\mathcal{G}$ and with $d_i$ the degree of the vertex $i \in \mathcal{V}$, hence the following upper bounds (UB).

\begin{itemize}
    \item Greedy coloring UB\cite{mitchem1976various}: $\chi(\mathcal{G}) \leq \Delta(\mathcal{G})+1 = UB_G$
    \item Welsh-Powell's UB\cite{ub_ch_num}: $\chi(\mathcal{G}) \leq \max\limits_{i \in \mathcal{V}} (min(d_i+1, i)) = UB_{WP}$, provided that $d_1 \geq d_2 \geq \cdots \geq d_n$ 
\end{itemize}

Let $\lambda_1, \ldots, \lambda_n$ be the eigenvalues of $A$, with $\lambda_1\geq\lambda_2\geq \cdots \geq \lambda_n$, and $(n^+,n^0,n^-)$ the inertia of $A$, \textit{i.e.}, the number of positive, null and negative eigenvalues respectively, then the following lower bounds (LB) hold.

\begin{itemize}
    \item Hoffman's LB\cite{hoffman}: $\chi(\mathcal{G}) \geq 1-(\lambda_1/\lambda_n) = LB_H$ 
    \item Elphick-Wocjan's LB\cite{elphick2016inertial}: $\chi(\mathcal{G}) \geq 1 + max(\frac{n^+}{n^-}, \frac{n^-}{n^+})=LB_{EW} $ 
    \item Edwards-Elphick's LB\cite{edwards1983lower}: $\chi(\mathcal{G}) \geq n/(n-\lambda_1)=LB_{EE}$
\end{itemize}

Putting together these concepts, \textbf{BBQ-mIS} represent an outer (with respect to the \textit{MIS} solver) optimization phase, which considers all solutions contained in the histogram coming out from the \textit{MIS} solver and branches over the induced graphs, obtained by removing all the vertexes in one of the parent node's \textit{mIS}s.

At the root of the BB tree, we start with the graph $\mathcal{G}$ to be colored. The BB scheme, as represented in Fig. \ref{fig:bb}, governs the optimization so that at the end of the procedure we can retrieve as the \textbf{BBQ-mIS} solution the best coloring found so far, \textit{i.e.}, the one that requires fewer colors.

At each BB node, the \textit{GC} problem is associated with a state (see Fig. \ref{fig:bb_node}), fully described by a tuple $(\mathcal{H},C,k,LB,UB)$ where
\begin{itemize}
    \item $\mathcal{H}$ is the induced subgraph obtained from the one in the parent node's state by removing the \textit{mIS} solution from which the branch springs up. $\mathcal{H}$ will be the input graph to the \textit{MIS} solver for the considered BB node off-springs;
    \item $C$ is a feasible coloring; if the BB node is at depth $l$ ($l=0$ at the BB root) of the BB tree and $\mathcal{H}=\mathcal{H}(\mathcal{U},\mathcal{I})$, $m=|\mathcal{U}|$, $C=\{C_1=mIS_1,C_2=mIS_2,\ldots C_l=mIS_l,C_{l+1}=u_1,\ldots,C_{l+m}=u_m\}$, $u_i \in \mathcal{U}$, \textit{i.e.}, the first $l$ colors are the ones inherited from the parents' \textit{mIS}s and the other $m$ colors are assigned according to the worst-case scenario (one color for each of the remaining vertexes in $\mathcal{H}$);
    \item $k$ is the value of the objective function in the BB node, \textit{i.e.}, the number of colors in $C$;
    \item $LB$ is a lower bound on $\chi(\mathcal{G})$ by pursuing the exploration of the BB node, \textit{i.e.}, for a node at depth $l$, $LB=l+$ lower bound on $\chi(\mathcal{H})$; this last lower bound is selected as the tightest among Hoffman, Elphick-Wocjan and Edwards-Elphick lower bounds, that is the $max(\lfloor LB_H\rfloor, \lfloor LB_{EW}\rfloor, \lfloor LB_{EE}\rfloor)$;
    \item $UB$ is an upper bound on $\chi(\mathcal{H})$, $UB = min(\lceil UB_G\rceil, \lceil UB_{WP}\rceil)$
\end{itemize}

Not all the solutions provided by the \textit{MIS} solver give birth to branches. In fact, the traditional pruning criteria of the BB method are complemented by some others which arise from the problem at hand. The effect of the following pruning criteria (see Fig. \ref{fig:bb_prune}) significantly reduce the number of BB nodes that are generated or explored.

\begin{itemize}
    \item \underline{Pruning by non-improving solution}: it applies when in the state of a BB node $LB \geq$ best value of the objective function found so far; this is a standard pruning criterion, which prevents us from wasting time and resources by exploring the potential children of a BB node which for sure would not provide a coloring with a smaller number of colors.
    \item \underline{Pruning by unfeasibility}: it avoids generating children from the \textit{MIS} solver which are either non-independent sets (they would result in some adjacent vertexes colored with the same color) or non-maximal sets (they would violate Th. \ref{th:mIS} on which the overall BB approach relies).
    \item \underline{Pruning by redundancy}: it avoids duplicating an exploration in the BB tree; if the same subgraph $\mathcal{H}$ has been generated previously, it will not generate a new BB node. Since the same induced graph can be obtained by removing the same vertexes, but in a different order, we associate with each $\mathcal{H}(\mathcal{U},\mathcal{I})$ a fingerprint, \textit{i.e.}, $fp(\mathcal{H})=\sum_{u \in \mathcal{U}}2^i-1$, and generate a new BB node only if its fingerprint has not yet been detected.
\end{itemize}

\begin{figure}[ht]%
    \centering
    \subfloat[\centering Each node in the Branch\&Bound tree is associated with the induced graph $\mathcal{H}$, the corresponding coloring $C$, the number of colors  $k$, a lower bound $LB$ on the best coloring achievable by branching further on that node, and an upper bound $UB$ on $\chi(\mathcal{H})$.]{{\includegraphics[width=\columnwidth]{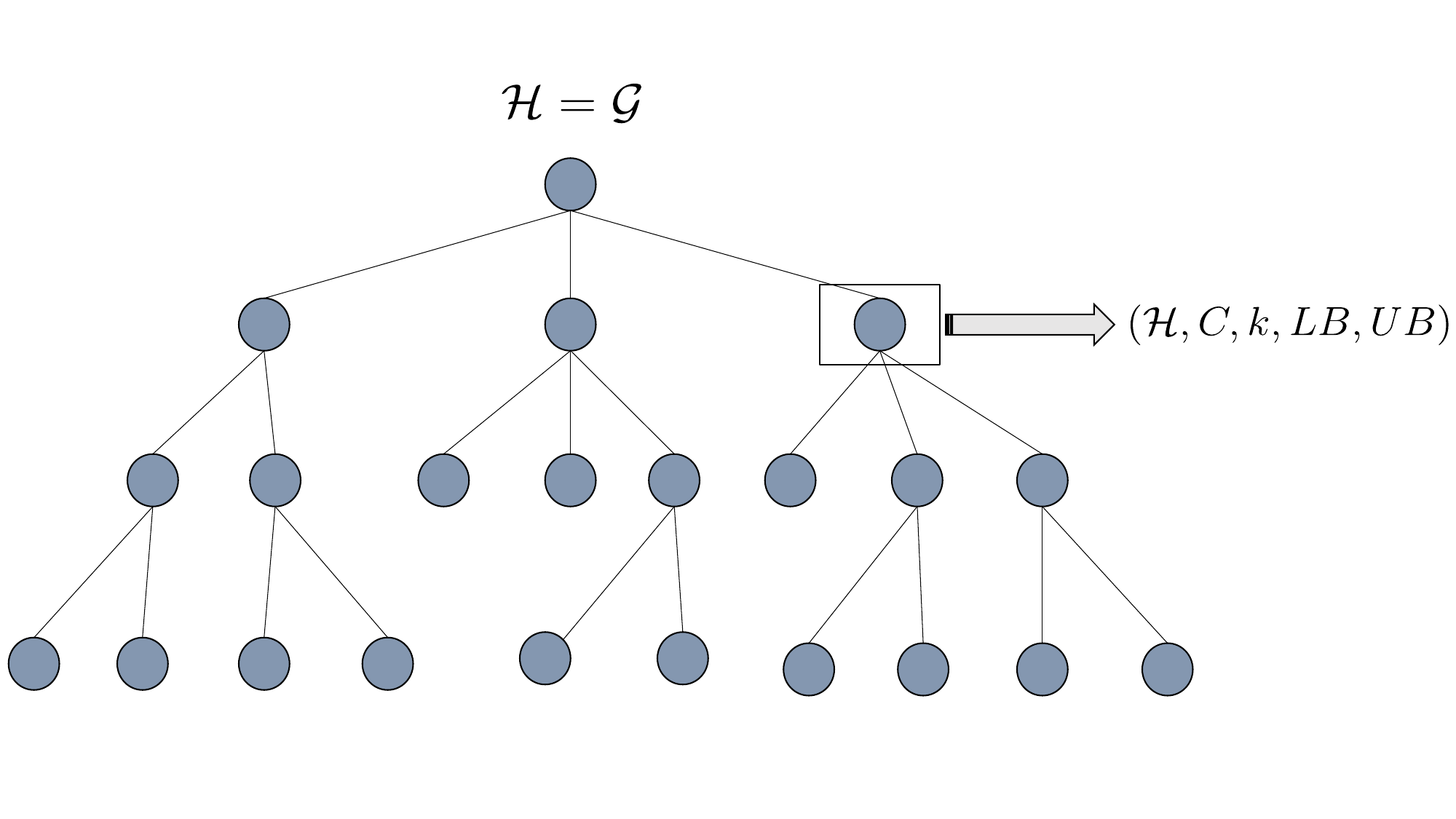}  \label{fig:bb_node}}}\\
    \subfloat[\centering The pruning criteria reduce the number of Branch\&Bound nodes to explore taking into account non-improving directions, unfeasibility, and equivalences of graphs $\mathcal{H}$.]{{\includegraphics[width=\columnwidth]{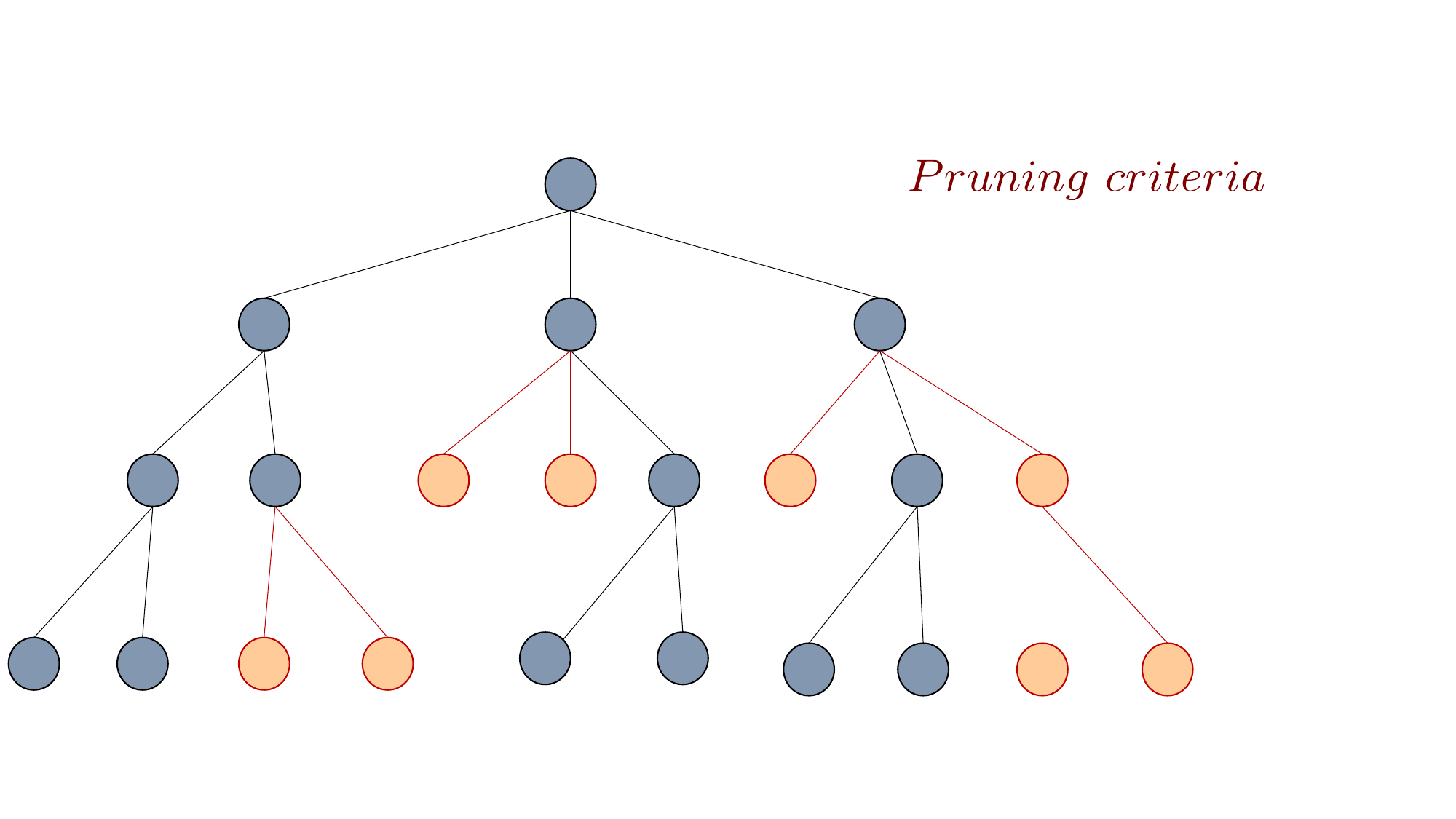}
    \label{fig:bb_prune}}}\\
    \subfloat[\centering A priority score is assigned to each Branch\&Bound node to lead the exploration of the Branch\&Bound tree: higher priority scores mean that the corresponding node is explored sooner.]{{\includegraphics[width=\columnwidth]{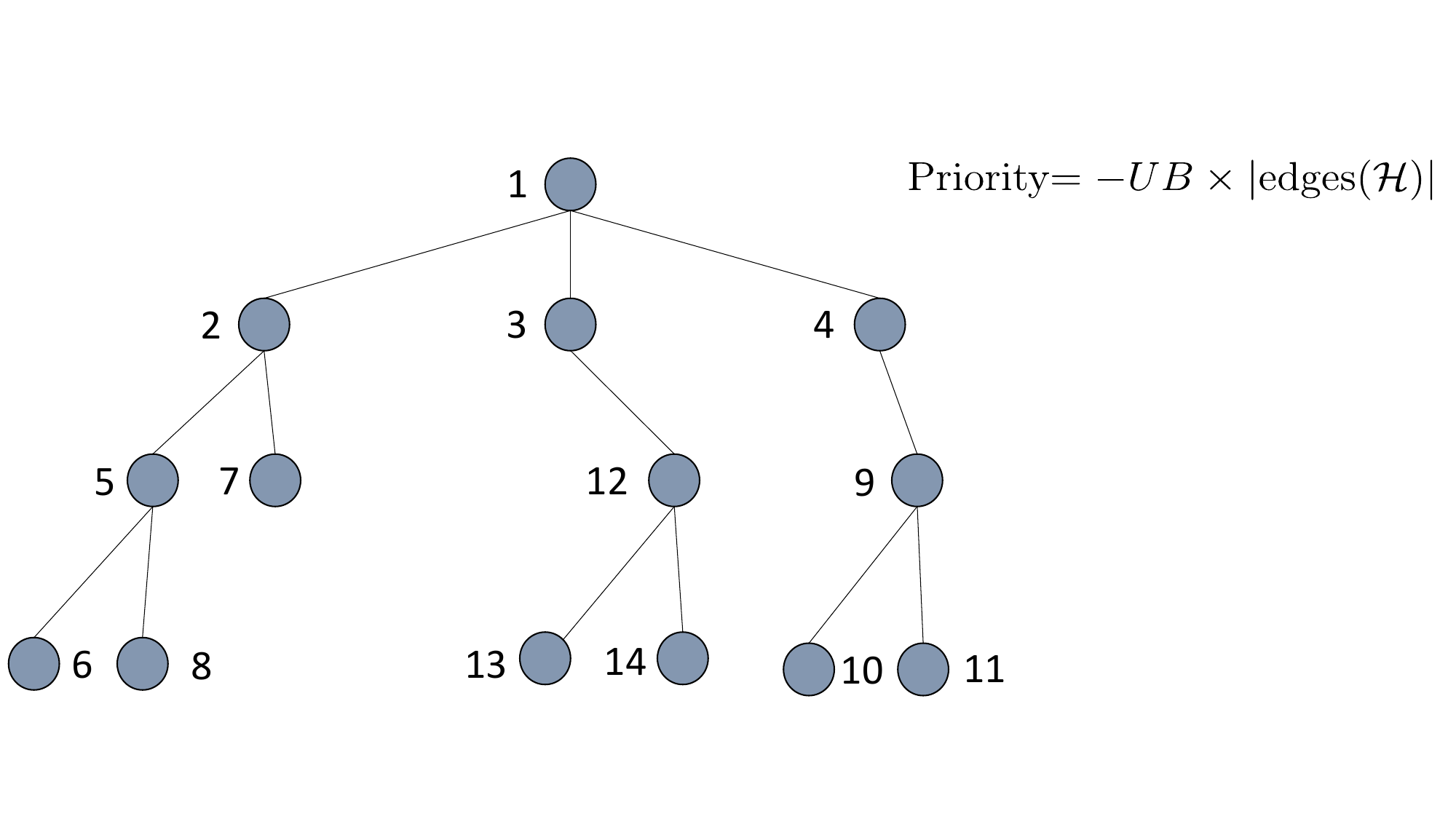}
    \label{fig:bb_exp}}}%
    \caption{Representation of the Branch\&Bound scheme underlying the \textbf{BBQ-mIS} algorithm.}%
    \label{fig:bb}
\end{figure}

Concerning the BB tree exploration, we first investigated standard policies, such as First-In-First-Out (FIFO), gap-based, and depth-first, then we designed our custom policy that assigns to each BB node a priority score (see Fig.\ref{fig:bb_exp}). The higher the priority score, the sooner the BB node is explored. To balance between the in-depth and in-width exploration, the priority score is computed as $-UB \times |\mathcal{I}|$. At the top of the BB tree, the priority scores are low. So, in the beginning, the BB exploration goes in-depth, thus providing a best objective function value similar to the one that can be obtained by the \textit{Greedy-it-MIS} algorithm. Then, it starts looking for the most promising directions in the BB tree, the ones with lower connectivity $|\mathcal{I}|$, combined with a tighter bound on the worst-case scenario provided by $UB$.

Finally, \textbf{BBQ-mIS}, as described so far, is designed to end with its solution once all the BB branches culminate into leaves. However, since the number of \textit{mIS}s scales linearly with the size of the graph $n$ \cite{byskov2004enumerating} and \textit{mIS}s are computed at each BB node, \textbf{BBQ-mIS} could explore a very large number of solutions before ending. Hence, we added a stopping criteria, after which the \textbf{BBQ-mIS} returns the best coloring found as its solution. In particular, we set a maximum number of BB nodes that may be explored to $50$.

\subsection{Quantum and classical solvers}\label{sec:mis_solver}

So far, we have focused mainly on the classical part of our hybrid algorithms, \textbf{Greedy-it-MIS} and \textbf{BBQ-mIS}. Here, we describe what is the \textit{MIS} solver that brings quantum into the aforementioned methods.

In our experiments, we relied on \textit{Pulser}\footnote{https://pulser.readthedocs.io} for the emulation of the quantum system on classical resources. Then, we adopted the Quantum Approximate Optimization Algorithm (QAOA) \cite{serret2020solving, pasqal_optim, farhi2014quantum} to optimize the laser pulses' parameters to solve an \textit{MIS} problem. Specifically, the input graph $\mathcal{G}$ is depicted on the array of rearranged Rydbger atoms. The rabi frequency $\Omega$ is set according to the desired rydberg radius $r_b$ so that the edges in $\mathcal{G}$ reproduce the proper connectivity. The detuning and the shape of the laser pulses are set according to the specifics of the real hardware \textit{Pasqal’s R\&D prototype Chadoq2} \cite{silverioPulserOpensourcePackage2022}. The maximum value for rabi frequency and detuning is $12MHz$, and the duration of the pulse sequences is limited to $3 \mu s$.

The parameters to optimize, \textit{i.e.}, the variational parameters in QAOA, are the time durations ($\mu s$) of the pulses. In this work, we considered just one layer of alternating non-commutative Hamiltonian. The durations' optimization relies on a classical solver (Nelder-Mead simplex algorithm\footnote{https://docs.scipy.org/}) that tunes them in order to minimize the energy of the Hamiltonian, so pushing the results towards \textit{MIS}s solutions.
Finally, thanks to the Rydberg blockade effect and the tuned laser pulses, we take repeated measurements and obtain a set of viable \textit{MIS}s.

This emulative process can be adopted for graphs with a limited number of vertexes, as it requires solving the differential equation describing the overall quantum system. Practically, it is possible to emulate quantum \textit{MIS} solution on graphs with up to $15-17$ vertexes within $3-20$ hours of emulation.


\textit{Gurobi} classical solver provides a reference value for the \textit{GC} problem solutions. Since the GC problems for which it was possible to emulate the quantum dynamics were pretty small, up to $15$ vertexes, \textit{Gurobi} solver always return an optimal coloring solution.

\subsection{Remarks about HPC-QC interplay}
The application described in this paper represents a good example of a hybrid algorithm for solving combinatorial optimization problems which can leverage both the computational power of modern HPC systems and the characteristics of some types of QPUs, in this case, neutral atoms-based ones. In fact, the BB is well-suited for parallel computing (our simple implementation already makes use of MPI to launch processes on multiple HPC processes) as each BB node is independent of the others, leading to batches of parallel computation where a set of the BB nodes to be explored is evaluated in parallel. In this context, the HPC produces a large number of quantum-friendly problems that can be solved using a QPU, inducing a few technical challenges related to the interplay between quantum computers and classical computers. Given the early stage of the NISQ platforms, a similar type of usage from typical users (mostly academic and research personnel) is foreseeable \cite{humbleQuantumComputersHighPerformance2021}.

The main challenge is related to parallelism: in the short term, most HPC facilities are expected to have at most one QPU for a given technology (\textit{e.g.}, neutral atoms) available: this represents a first bottleneck for this class of problems where many quantum sub-problems are generated; key parameters, in this case, are related to the relative length of the quantum and classical task and to the latency of moving data around. In this specific case, the quantum task requires 50 \emph{shots}, that are executed by the QPU at a rate that we can assume to be $\sim 5Hz$ \cite{wurtzAquilaQuEra256qubit2023, QAOAQAASolve2023}. This means that sampling an \textit{MIS} requires around $10s$, regardless of the size of the graph; also, \textit{MIS}s can only be addressed one at a time, or a few at a time if considering approaches like mapping multiple graphs on the same QPU register. This granularity is much larger than the time taken by the classical computer to generate problems, creating a bottleneck and wasting all of the advantages of parallelizing the BB tree exploration. On the other hand, the impact of data movement is negligible at this granularity.

The shot rate is then a critical factor to keep into account for an efficient HPC-QC integration. Another related aspect is the scheduling of resources: a low shot rate keeps classical resources idle for a long while waiting for the QPU to complete its task, resulting in a sub-optimal allocation. Possibilities to mitigate this issue range from oversubscribing classical resources when hybrid jobs are involved, to splitting the calculation into multiple steps (\textit{i.e.}, expressing it as a workflow), decoupling the different classical and quantum steps into separate units of scheduling, allowing the batch scheduler to allocate resources to them only when they are actually ready to run. This last approach is particularly suitable also when dealing with cloud-based quantum resources, or with mixed set-ups.

Finally, assuming a different technology (\textit{e.g.}, superconducting) or a dramatically improved shot rate, other resource scheduling challenges should be considered: with only a few QPUs available, we can imagine having a large number of small quantum tasks to be submitted by multiple classical jobs executed by many users. Efficient task queues should be implemented to handle these tasks at a much smaller granularity \cite{schulzAcceleratingHPCQuantum2022a}.

\section{Results and Discussion}

\begin{figure*}[htbp]
    \includegraphics[width=\linewidth]{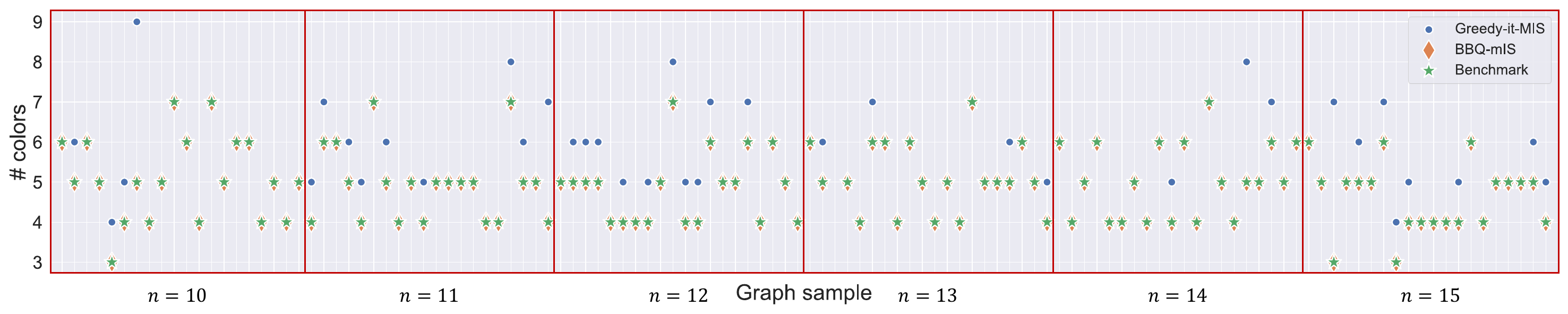}
    \caption{Comparison of \textit{GC} results: \textbf{Greedy-it-MIS} and \textbf{BBQ-mIS} algorithms are used to color all graphs in our dataset, \textit{Gurobi} solver provides a benchmark for the coloring solution. The implementation for the \textit{MIS} solver is based on the \textit{Pulser} library.}\label{fig:compQ}
\end{figure*}

To test our algorithms, we created a dataset containing $120$ samples of unit-disk graphs with a number of vertexes $n \in \{10,11,12,13,14,15\}$. Specifically, there are $20$ samples for each possible value of $n$. The unit-disk feasible representation of each graph is obtained with the \textit{GEAN model} presented in \cite{embedding_links}.

All the experiments were run on an IBM Power9-based cluster, with 32 cores/node and 256 GB/node.  In particular, each trial uses 32 physical cores of a node and used up to 16 MPI processes for the parallelization of the Branch\&Bound. The overall experiment required $\approx$ 8000 core hours.

When solving \textit{GC} problems, the \textit{MIS} solver emulates the quantum system, thanks to the \textit{Pulser} library implementation, as discussed in section \ref{sec:mis_solver}.

On this dataset, the \textbf{BBQ-mIS} algorithm always reaches the coloring solutions with the same number of colors as the optimal \textit{GC} solutions. \textbf{Greedy-it-MIS} instead provides worse solutions for $38$ samples out of $120$; in the worst case, it requires up to $4$ colors more ($9$ colors instead of $5$). Fig. \ref{fig:compQ} summarizes all \textit{GC} results on our dataset.

Another observation concerns the number of nodes used in the Branch\&Bound exploration. In the case of our dataset, the maximum number of BB nodes, \textit{i.e.}, $50$, is never reached: \textbf{BBQ-mIS} terminates into leaves after having explored at most $20$ nodes, in the worst-case graph instances. Reasonably, the number of \textit{mIS} coming from small graphs, hence the number of children generating from each branch, is significantly lower than when targeting larger graphs.

When moving to larger graphs and simulations with real quantum hardware, it might be interesting to modify the computation of the feasible coloring $C$ with a less greedy approach than the one which assigns one different color to each of the remaining vertexes in $\mathcal{H}$. Thus, one could identify a property on the number of vertexes $|\mathcal{U}|$ to establish that when the graph is smaller than a given threshold it is better to solve the subgraph coloring with exact classical methods, possibly even enhanced by HPC, whereas for graph larger than the threshold, the quantum \textit{MIS} solver can be exploited to provide viable \textit{mIS} for partial coloring in a computationally efficient way.

Finally, given the current version of \textbf{BBQ-mIS}, and the sampling rate of the quantum machine, we estimated the time needed to solve a GC, independently of the graph size. To provide meaningful statistics, we need $50$ samples for a fixed pulse setting. This sampling procedure is a minimal quantum task. To optimize the pulse durations, we allow for evaluating the sampling results $100$ times. Once the optimizer finds the best duration values, we retrieve $100$ samples to provide the final result of the \textit{MIS} solver. This procedure takes place each time a BB node is explored. Since we limited our optimization to $50$ BB nodes in the worst-case scenario, we would need $(50\times100+100)\times 50=255000$ samples to solve a GC problem instance. Considering a $5Hz$ sampling rate, the corresponding computational time is $\approx 14$ hours. For the specific graph instances of our dataset, this time would be lower because the BB ended by optimality before exploring $50$ nodes. In particular, all the corresponding GC problems were solved with $8-20$ BB nodes, which would correspond to $2-6$ hours on the real QPU.
\section{Conclusion}

The work presented so far shows promising results for Graph Coloring problem solutions. Nevertheless, some further directions can be explored to improve \textbf{BBQ-mIS}'s performances or provide a broader benchmark.

One possibility is considering the Quantum Adiabatic Algorithm~\cite{albash2018adiabatic} to solve the \textit{Maximum Independent Set} problem instead of the QAOA approach, thus also reducing the use of quantum resources because no pulse optimization phase would be needed.

On the emulation side, it may be interesting to implement the same algorithms exploiting other libraries that emulate the neutral atoms quantum system, such as \textit{Bloqade}\footnote{https://queracomputing.github.io/Bloqade.jl}, potentially being able to target even larger graphs. Concerning this aspect, we aim to compare the emulated results with the simulations on real quantum hardware, thus establishing the robustness of our methods in the presence of noise.

Finally, even though these hybrid quantum-classical algorithms target \textit{GC} problems, they can inspire methods with a main focus on saving the required qubits and algorithmic integration with HPC systems, for solving other combinatorial optimization problems. 

\section*{Acknowledgment}
We would like to thank the CINECA consortium for the support and the provisioning of computational resources on the Marconi100 cluster.


\bibliographystyle{IEEEtran}
\bibliography{bbq.bib}

\end{document}